\documentclass[12pt,reqno]{amsart}
\usepackage[utf8]{inputenc}
\usepackage{color}
\usepackage{dutchcal}
\usepackage[toc]{appendix}
\usepackage{amsfonts,amsmath,centernot}
\usepackage{latexsym}
\usepackage{stmaryrd}
\usepackage{amssymb}
\usepackage{amsmath,amssymb,amsthm,hyperref,mathrsfs}
\usepackage{graphicx}
\usepackage{times}
\usepackage{fullpage}
\usepackage{latexsym}

\makeatletter
\newsavebox{\@brx}
\newcommand{\llangle}[1][]{\savebox{\@brx}{\(\m@th{#1\langle}\)}%
  \mathopen{\copy\@brx\kern-0.5\wd\@brx\usebox{\@brx}}}
\newcommand{\rrangle}[1][]{\savebox{\@brx}{\(\m@th{#1\rangle}\)}%
  \mathclose{\copy\@brx\kern-0.5\wd\@brx\usebox{\@brx}}}
\makeatother

\newtheorem{Theorem}{Theorem}[section]
\newtheorem{Proposition}[Theorem]{Proposition}
\newtheorem {Cor}[Theorem]{Corollary}
\newtheorem {pro}[Theorem]{Proposition}

\newtheorem {Hypothesis}[Theorem]{Hypothesis}
\newtheorem {Lemma}[Theorem]{Lemma}
\newtheorem {rem}[Theorem]{Remark}
\newtheorem*{example*}{Example}

\newtheorem {rems}[Theorem]{Remarks}
\newtheorem {com}[Theorem]{Comment}
\newtheorem {coms}[Theorem]{Comments}
\newtheorem {notation}[Theorem]{Notation}
\newtheorem {Definition}[Theorem]{Definition}

\newcommand{\bcom}{\begin{com} \rm } \newcommand{\ecom}{\end{com}}
\newcommand{\bcoms}{\begin{coms} \rm } \newcommand{\ecoms}{\end{coms}}
\newcommand {\bdef}{\begin{Definition}}
\newcommand {\edefi}{\end{Definition}}
\newcommand {\bl}{\begin{Lemma}}
\newcommand {\el}{\end{Lemma}}
\newcommand {\bethe}{\begin{Theorem}}
\newcommand {\eethe}{\end{Theorem}}
\newcommand {\bp}{\begin{pro}}
\newcommand {\ep}{\end{pro}}
\newcommand {\bcor}{\begin{Cor}}
\newcommand {\ecor}{\end{Cor}}
 \newcommand {\brem }{\begin{rem} \rm }
\newcommand {\erem }{\end{rem}}
 \newcommand {\brems }{\begin{rems} \rm }
\newcommand {\erems }{\end{rems}}

\newcommand{\ba}{\begin{array}}
\newcommand{\ea}{\end{array}}

\newcommand {\be}{\begin{equation}}
\newcommand {\ee}{\end{equation}}
\newcommand {\bde}{\begin{displaymath}}
\newcommand {\ede}{\end{displaymath}}
\newcommand {\beq}{\begin{eqnarray*}}
\newcommand {\eeq}{\end{eqnarray*}}
\newcommand {\beqa}{\begin{eqnarray}}
\newcommand {\eeqa}{\end{eqnarray}}

\def\proof{\noindent {\it Proof. $\, $}}
\def\finproof{\hfill $\Box$ \vskip 5 pt}
\def \ind{1\!\!1}

\newcommand{\bG}{\mathbb{G}}
\newcommand{\bF}{\mathbb{F}}

\newcommand{\bP}{\mathbb{P}}
\newcommand{\bQ}{\mathbb{Q}}

\newcommand{\E}{\mathbb{E}}

\newcommand{\F}{\mathcal{F}}
\newcommand{\G}{\mathcal{G}}

\newcommand{\nhookrightarrow}{\centernot{\hookrightarrow}}

\begin{document}

\title[]{Insiders and their Free Lunches: the Role of Short Positions}

\author{Delia Coculescu}
\address{University of Zurich\\ Department of Banking and Finance, Plattenstrasse 32\\ Z\"{u}rich 8032, Switzerland.}
\email{delia.coculescu@bf.uzh.ch}
\author{Aditi Dandapani}
\email{ad2259@caa.columbia.edu}
\begin{abstract}Given a stock price process, we analyse the potential of  arbitrage in a context of short-selling prohibitions.  We introduce the notion of minimal supermartingale measure, and we analyse its properties in connection to the minimal martingale measure. This question is more specifically analysed in the case of an investor having additional, inside information. In particular, we establish conditions when  minimal martingale and supermartingale measures both fail to exist.  These correspond to the case when the insider information includes some non null events that are perceived as having null probabilities by the uninformed market investors, even as they cannot observe them. The results may have different applications, such as in problems related to the local risk-minimisation for insiders whenever strategies are implemented without short selling. 

 \end{abstract}
\maketitle
\section{Introduction}

Given a stock price process, we analyse the potential of  arbitrage by insiders, when they chose not to sell short. Short-selling stocks when having access to  a value-destroying event before it is known to the public leads to profits, but insiders that are already detaining the stocks may profit without engaging in selling short.  Selling (and buying) of company shares by people who have special information because they are involved with the company is illegal in many jurisdictions, whenever the information is non-public and material. 
Insider trading scandals are nevertheless regularly on the front page of newspapers.
For instance, the US justice department is currently investigating alleged insider trading by US lawmakers who sold stocks just before the coronavirus pandemic sparked a major market downturn. Rarely short selling is part of the insiders' strategies in these scandals. Short selling transactions receive additional attention from supervisory authorities in their attempt to detect abnormal trading activity and it is possible that many insiders refrain from such strategies for their trades to remain undetected.   It is not clear nevertheless, from both a practical and purely theoretical perspective if an insider that does not detain the stock has the possibility to first buy and then sell the stock in a profitable manner, thus exploiting the informational advantage without selling short.   

The aim of this paper is to propose some first elements in tackling this question. For the notion of no arbitrage profits, we consider the setting of no free lunch with vanishing risks with short sales prohibitions (NFLVRS) introduced in Pulido \cite{Pulido14},  a restriction of the classical notion of no free lunch with vanishing risks (NFLVR) by Delbaen and Schachermeyer \cite{DelbScha94}, \cite{DelbScha98}.  Based on  previous work by Jouini and Kallal \cite{JouiKall95}, Fritelli \cite{Fritelli97}, Pham and Touzi \cite{PhamTouz99}, Napp \cite{Napp03} and Karatzas and Kardaras \cite{KaraKard07}, the paper by  Pulido  \cite{Pulido14} establishes important properties of price processes under short sale prohibitions namely the equivalence between (NFLVRS) and the existence of an equivalent supermartingale measure for the price processes. Structure conditions for a price process under (NFLVRS) are provided in Coculescu and Jeanblanc \cite{CocJean19}. Recently, Platen and Tappe \cite{PlattenTappe} raised a similar question to the one we are analysing here, namely if exploiting arbitrages requires short-selling, however without insider information being specifically analysed. It is shown that without any possibility of selling short, it is not possible to exploit arbitrages. But at the difference of their paper, we allow the investor to borrow risk-free (that is equivalent selling short risk-free assets), and this assumption results in additional possibilities for exploiting possible arbitrages. Hence, the results in \cite{PlattenTappe} do not apply to our framework. 

 There is a very vast literature studying arbitrage possibilities and characterising no-arbitrage models, for different types of  inside information and for different notions of arbitrage, sometimes weaker that no free lunch with vanishing risks. Usually,  a model with respect to a given information set (filtration)  is specified  on a probability space, representing the financial market such as common investors perceive it, and the insider information is introduced as a strictly larger information set, via an enlargement of the original filtration. A fundamental question in this context is  whether the additional information allows for arbitrage profits. 
We mention the early work by Grorud and Pontier \cite{GrorudPontier},  Amendinger et al. \cite{AmeImkSch98}, Imkeller \cite{Imkeller02} and subsequent work by  Hillairet \cite{Hill05}, Acciaio et al. \cite{AcciFontKard16}, Aksamit et al. \cite{AksaChouJean15}, Fontana et al \cite{FontJeanSong14}, Chau et al \cite{ChauTank18}, among athers.
 In a general setting, the situations where the additional information of an insider is insufficient for obtaining profits of arbitrage in the form of no free lunch with vanishing risk,  were characterised in the literature and linked to the so-called (H) hypothesis holding between the asset price filtration and the larger filtration containing the additional information, under an equivalent martingale measure (see Blanchet-Scalliet and Jeanblanc \cite{bsjeanblanc} and Coculescu et al. \cite{CocJeanNik12}).  
    
 The present paper contributes to the literature on insider trading by  studying specifically the short-sales restrictions.   
Within the two-filtration setting above described, with a filtration larger than the other, and where a stock price is an adapted process for both filtrations, we shall  assume that no free lunch with vanishing risks prevails for the common investors having access only to the smallest filtration. We will then exclusively focus on the no arbitrage conditions in the larger filtration, that is, when considering the insider information. As mentioned above, the existence of a supermartingale measure for the stock price is equivalent to (NFLVRS) (Pulido \cite{Pulido14}). After we introduce the notion of minimal supermartingale measure, we analyse some of  its connections with the minimal martingale measure. The results in Section \ref{Sec3} are actually general  and do not depend on the two-filtration setting. In Section \ref{secH}, a further analysis of the connections between the two measures is provided, specific to the the case of an insider, when assuming the smaller filtration has the predictable representation property with respect to a continuous martingale.  The failure of both the minimal martingale measure and the minimal supermartingale measure to exist is obtained when  the (H) hypothesis is satisfied under a probability measure such that the original one is absolutely continuous to it, but not equivalent. The additional null sets of the original measure belong to the inside information: they are not measurable in the smaller filtration that models the information of the common investors. Some new results on the importance of the (H) hypothesis may also be found in Section \ref{secH}, in particular for the case where (NFLVRS) holds but  (NFLVR) fails to hold.  This corresponds to the case where the arbitrages cannot be exploited without selling the stock short. In the Conclusions section, we shall discuss some possible extensions and uses of our results.

\section{The mathematical setting}\label{SectionSetting}
A probability space $(\Omega, \G, \bP )$ is fixed, endowed with a right-continuous filtration $\bF^0=(\F_t^0)_{t\geq 0}$.  We denote by $\bF=(\F_t)_{t\geq 0}$ the augmentation of  $\bF^0$ with the null sets of $\bP$,  ensuring that $\bF$ satisfies the usual assumptions.  We consider a financial market where a risky asset (e.g. a stock) is traded, in addition to a risk-free asset and such that  the price process, denoted by $X=(X_t)_{t\geq 0}$ is adapted.  To keep the message simple, we assume that the price of the risk-free asset is constant. Alternatively, one could consider the discounted asset price.   Through this paper, we make the following assumption:

\begin{Hypothesis}\label{H1} There exists an equivalent local martingale measure for $X$ in the filtration $\bF$. \end{Hypothesis} 

Hypothesis \ref{H1} more precisely means:  there exists a probability $\bP^*\sim \bP$ on $(\Omega,\F_\infty)$ such that $X$ is a $(\bF,\bP^*)$ local martingale. 
We emphasise that it is only required the equivalence of these measures on $\F_\infty$. Implicit to Hypothesis \ref{H1} is also that  $X$ is a  locally bounded $(\bF,\bP)$ semimartingale. We have thus a standard arbitrage free market in the sense of no free lunch with vanishing risks \cite{DelbScha94}, whenever the filtration $\bF$ summarises the available information. We shall call \textit{common investors} the $\bF$ informed investors.

We now take the point of view of a particular agent, that we name the \textit{insider}. We assume that the insider has access to some information that can potentially impact the risky asset price, should this information be publicly released. Hence, the insider may have access to some forms or arbitrage, even under Hypothesis \ref{H1}, as the asset price does not necessarily incorporate the extra information.  In mathematical terms, there is a right continuous filtration $\bG^0$ strictly larger than $\bF^0$, i.e., such that 
\begin{equation}\label{fg0}
\F^0_t\subset \G^0_t\;, \;t\geq 0\text{ and }\bF^0\neq \bG^0.
\end{equation}We define $\bG$ as the $\bP$ augmentation of $\bG^0$, so that $\bG$ satisfies the usual assumptions.  The insider is assumed to be able to use $\bG$ adapted strategies in order to make investments in $X$ and the risk-free asset. We do not assume any particular structure of the filtration $\bG$, except being strictly larger than $\bF$, in the sense of (\ref{fg0}). The filtrations $\bF^0$ and $\bG^0$ will still play a role later on, as we consider probabilities to which $\bP$ is only absolutely continuous. For now, the framework is that of the probability space $(\Omega,\G,\bP)$ endowed with two filtrations $\bF\subset\bG$.

Importantly, we shall assume a certain structure for the decomposition of the stochastic process $X$ in the insider filtration $\bG$, namely:
\begin{Hypothesis}\label{H2}  The price of the risky asset $X$ is a nonnegative $(\bG,\bP)$ semimartingale  of the form: 
\be\label{Xdecomp}
X=X_0+M+ \int \alpha d\langle M\rangle
\ee
for some $(\bG,\bP)$ local martingale $M$ and a $(\bG,\bP)$ predictable process $\alpha$. The process $\langle M\rangle$, the sharp bracket  of $M$,  is  the $(\bG,\bP)$ dual predictable projection of the quadratic variation of $M$. 
For simplicity, and without loss of generality, we shall assume $\alpha$ to be constantly zero on the intervals where $d\langle M\rangle=0$.
\end{Hypothesis}

Before considering the question of whether the insider can have free lunches, one should clarify the available trading strategies for the insider. The insider is supposed to invest in the risky and risk free asset based on her information. We are interested here in two alternative situations: the one where the insider may sell short the risky asset and the one where she is constrained to have only long positions in the risky asset, i.e., short sales are prohibited. Consequently, we define the following sets for $t\geq0$:
\begin{align*}
\mathcal M_t(\bG,\bP):&=\{\bQ\sim\bP\ \text{ on }\G_t\;|\; X_{\cdot\wedge t} \text{ is a $(\bG,\bQ)$ local martingale}\} \text{ and }\mathcal M(\bG,\bP):=\mathcal M_\infty(\bG,\bP);\\
\mathcal{S}_t(\bG,\bP):&= \{\bQ\sim\bP \text{ on }\G_t\;|\; X_{\cdot\wedge t}  \text{ is a $(\bG,\bQ)$ supermartingale}\}\text{ and }\mathcal{S}(\bG,\bP):=\mathcal{S}_\infty(\bG,\bP).
\end{align*}

Obviously, we have that
\[
\mathcal M_t(\bG,\bP)\subset \mathcal{S}_t(\bG,\bP)\quad t\geq 0.
\]

In this paper we intend to analyse from the insider's perspective the connections between arbitrages involving short selling the risky asset and arbitrages that exclude short selling the risky asset.  The no arbitrage conditions we study are relative  to the insider filtration $\bG$, and they are formally the following:
\begin{itemize}
\item[(a)] the no free lunch with vanishing risks (NFLVR) condition of Delbaen and Schachermayer \cite{DelbScha94},  in the case where short positions are permitted along with the long positions in the assets. In our setting, (NFLVR) is equivalent to: $\mathcal M(\bG,\bP)\neq\emptyset$. 
\item[(b)] the no free lunch with vanishing risk under short sales prohibition (NFLVRS), a condition that was introduced in Pulido \cite{Pulido14}. In our setting,  by short sales prohibitions we mean that long positions in the risky asset are permitted, along with short and long in the risk-free asset, and in this case (NFLVRS) is equivalent to: $\mathcal{S}(\bG,\bP)\neq\emptyset$ (cf. Theorem 3.9 in \cite{Pulido14}). 
\end{itemize}
   
\begin{notation}\label{notation}Given a $\bG$-adapted process $Y=(Y_t)_{t\geq 0}$, we use the following notation
\begin{itemize}
\item[-] $\;^{o,\bP}Y$ stands for the optional projection of the process $Y$ onto the filtration $\bF$, with respect to the measure $\bP$;
\item[-]  whenever $Y$ is  increasing: $Y^{o,\bP}$ stands for its $(\bF,\bP)$-dual optional projection; $Y^{p,\bP}$ stands for its $(\bF,\bP)$-dual predictable projection;   $\{dY> 0\}$ stands for the support of the random measure $dY$.  Hence, there exists a unique subset of $\Omega\times \mathbb R_+$ that we denote   $\{dY=0\}$ and satisfying:  $\{dY=0\}\cup\{dY>0\}=\Omega\times \mathbb R_+$ and  $\{dY=0\}\cap\{dY>0\}=\emptyset$. 
\item[-] whenever $Y$ is a local martingale, we denote by $\langle Y\rangle$ its sharp bracket with respect to $\bG$. If $Y$ is an $\bF$ adapted local martingale, we denote by $\llangle Y\rrangle$ its sharp bracket with respect to $\bF$. In other words, we have $\llangle Y\rrangle=([Y])^{p,\bP}$. We emphasise that whenever $Y$ is continuous,  we simply have: $\llangle Y\rrangle= \langle Y\rangle=[Y]$.
\item[-] whenever $Y$ is a semimartingale with $Y_0=0$, $\mathcal E (Y)$ denotes the Dol\'eans-Dade exponential.
\end{itemize}
\end{notation}

Some results will be established in the particular settings where the process $X$ has continuous sample paths, or $\bF$ enjoys the predictable representation property, that is, we have a complete $\bF$ market. 
We label these particular settings as follows:
\begin{itemize}
\item[(C)] The stochastic process $X$ has continuous sample path.
\item[$\bF$-(PRP)] The filtration $\bF$ is the natural filtration of the price process $X$. There is a (one dimensional) $(\bF,\bP)$ local martingale $m^{\bF,\bP}$ such that $m^{\bF,\bP}$ satisfies the $(\bF,\bP)$ predictable representation property, that is: any $(\bF,\bP)$ local martingale vanishing at zero is equal to $H\cdot m^{\bF,\bP}$ for a suitable $\bF$ predictable process $H$.
\end{itemize}

\bigskip

In the remaining of this section, we emphasise some implications of our modelling assumptions. 
The choice of assuming  a decomposition of $X$ as proposed in (\ref{Xdecomp}) is rooted in the result below:
\begin{Theorem}\label{ThmXdecomp} Suppose that Hypothesis \ref{H1} holds. In addition, suppose that at least one of the following two conditions is satisfied:
\begin{itemize}
\item[1.] $\mathcal M(\bG,\bP)\neq\emptyset$
\item[2.] (C) and $\mathcal{S}(\bG,\bP)\neq\emptyset$.
\end{itemize}
Then,  Hypothesis \ref{H2} holds true. 
\end{Theorem}
\proof
In the case where $\mathcal M(\bG,\bP)\neq\emptyset$ holds, the decomposition in (\ref{Xdecomp}) is due to Ansel and Stricker \cite{AnselStick92}  and Schweizer \cite {Schw92}  for a $\bG$ predictable process $\alpha$. The same holds when  $\mathcal{S}(\bG,\bP)\neq\emptyset$, as proved by Coculescu and Jeanblanc \cite{CocJean19} (note that should $X$ be a discontinuous price process the decomposition above is not automatic). 
\finproof

From Theorem \ref{ThmXdecomp}, we may characterise a framework where the insider can have free lunches with vanishing risk without engaging in short sales:
\begin{Cor}
 Suppose that Hypothesis \ref{H1} and condition (C) hold. Then, if Hypothesis \ref{H2} is not satisfied, it follows that $\mathcal M(\bG,\bP)=\emptyset$ and $\mathcal{S}(\bG,\bP)=\emptyset$.
\end{Cor}
In Coculescu and Jeanblanc \cite {CocJean19}, examples are provided where the  $(\bG,\bP)$ decomposition of  $X$ fails to satisfy  (\ref{Xdecomp}), in the case where $\bG$ is obtained by  progressive enlargements of $\bF$ with a random time. The corresponding arbitrage portfolios are also discussed.  The current paper is intended to further analyse the question of (NFLVRS) within  the frame of Hypothesis \ref{H2}, that is, the representation (\ref{Xdecomp}) for the stock price holds.

Finally, we can  notice that implicit to Hypothesis \ref{H1} is also a $(\bF,\bP)$ decomposition of the price process $X$, as follows: there is a $(\bF,\bP)$ local martingale $m^{\bF,\bP}$ such that
\begin{equation}\label{XdecompF}
X_t=X_0+\int_0^ta_sd\llangle m^{\bF,\bP}\rrangle_s +m^{\bF,\bP}_t.
\end{equation}
where we have (recall our notation) $\llangle m^{\bF,\bP}\rrangle=\;^{o,\bP} [m^{\bF,\bP}]$, that is the sharp bracket of the  $(\bF,\bP)$ local martingale $m^{\bF,\bP}$.  This decomposition follows by the same argument as in the proof of Theorem \ref{ThmXdecomp} (i), but applied to $\bF$ instead of $\bG$.

Given Hypotheses \ref{H1} and \ref{H2}, it follows that  $m^{\bF,\bP}$ is a $(\bG,\bP)$ semimartingale that decomposes:
\[
m^{\bF,\bP}_t:=\left(\int_0^t\alpha_sd\langle M\rangle_s-\int_0^ta_sd\llangle m^{\bF,\bP}\rrangle_s\right) +M_t.
\]  

Further analysis in Section \ref{secH} will again give rise to the representations (\ref{Xdecomp}) and (\ref{XdecompF}) by means of Girsanov transformations. It is important to keep in mind that we assume that no free lunch with vanishing risk holds for common investors (that is in the smaller filtration $\bF$), and we analyse the properties of the price process when some additional information is available for an insider, that is, in the filtration $\bG$. 

\section{The minimal supermartingale measure: definition and some sufficient conditions for its existence}\label{Sec3}
In this section we introduce the concept of  minimal supermartingale measure and establish some connections with the minimal martingale measure.  We use the framework from the previous section, in particular the representation of the asset price in Hypothesis \ref{H2} is supposed to be verified.   However, Hypothesis \ref{H1} is not used in establishing the results in this section, so that the results are general and still hold even if it is were not satisfied. We introduce the following $(\bG,\bP)$ local martingales:
 \begin{align}\label{R}
R_t:&= \mathcal E_t\left(-\int_0^\cdot \alpha_ud M_u\right), \quad t\geq 0\\\label{Rplus}
R^{(+)}_t:&= \mathcal E_t\left(-\int_0^\cdot (\alpha_u)^{+}d M_u\right), \quad t\geq 0.
\end{align}
We denote
\[
\frac{d\bQ^{\mathcal M}}{d\bP}\Big|_{\G_t}=R_t,
\] and if $R$ is a strictly positive and uniformly integrable martingale, we have $\bQ^{\mathcal M}\in\mathcal M(\bG,\bP)$. Similarly, we denote
\begin{equation}\label{defQminsup}
\frac{d\bQ^{\mathcal S}}{d\bP}\Big|_{\G_t}=R^{(+)}_t,
\end{equation} and if $R^{(+)}$ is a strictly positive and uniformly integrable martingale we have $\bQ^{\mathcal S}\in\mathcal{S}(\bG,\bP)$ which can be checked using Girsanov's theorem. In general, $\bQ^{\mathcal M}$ and $\bQ^{\mathcal S}$ are signed measures.

We recall below the notion of  minimal martingale measure  which first appeared in Schweizer \cite{Sch88} (see also  F\"ollmer and Schweizer (\cite{FoelSchw10}) for a survey of this very rich topic).

\begin{Definition}
When $R$ is a strictly positive and square integrable $(\bG,\bP)$ martingale, we call $\bQ^{\mathcal M}\in\mathcal M(\bG,\bP)$ the $(\bG,\bP)$ minimal  martingale measure for $X$.
\end{Definition}
Symmetrically, let us introduce the following definition:
\begin{Definition}
When $R^{(+)}$ is a strictly positive and square integrable $(\bG,\bP)$ martingale, we call $\bQ^{\mathcal S}\in\mathcal{S}(\bG,\bP)$, that is defined as in (\ref{defQminsup}),  the $(\bG,\bP)$ minimal supermartingale measure for $X$.

\end{Definition}

In what follows, we will skip the specification of the filtration and of the reference probability $(\bG,\bP)$ and will be simply referring to $\bQ^{\mathcal M}$ resp. $\bQ^{\mathcal S}$ as  the minimal martingale (resp. supermartingale) measure. 

It can be shown that when $X$ has continuous sample paths, the minimal supermartingale measure can be characterised as the unique 
solution of a minimisation problem involving the relative entropy with respect to $\bP$, in a parallel manner to the result holding for the minimal martingale measure; in both cases the idea is to obtain a ``minimal modification'' of $\bP$ that allows to reach the set of supermartingale versus martingale measures. The formal setting is given below.
\begin{Definition}
Given two probability measures $\bP$ and $\bQ$ on $(\Omega,\G)$, the \textit{relative entropy}, $H(\bQ| \bP)$  is defined as
\begin{equation*}
  H(\bQ|\bP)=\begin{cases}
    \int \log \frac{d\bQ}{d\bP}d\bQ & \text{if $\bQ \ll \bP$}.\\
    +\infty, & \text{otherwise}.
  \end{cases}
\end{equation*}
\end{Definition}
The relative entropy is  interpreted as a directed divergence measure between two probabilities; it is nonnegative and $  H(\bQ|\bP)=0$ if and only if $\bQ =\bP.$

\begin{Theorem}\label{Entropy}
We assume that Hypothesis \ref{H2} and (C) hold. Assume also that $\bQ^{\mathcal S}\in \mathcal{S}_T(\bG,\bP)$ and $ H(\bQ^{\mathcal S}|\bP)<\infty$.
Then, the measure $\bQ^{\mathcal S}$ is the unique solution of 
\be\label{minentropy}
\text{Minimize }H(\bQ| \bP)-\frac{1}{2}\E^{\bQ}\left [\int_{0}^{T}(\alpha^{+}_s)^2d\langle X \rangle_s\right],
\ee
over all supermartingale measures $\bQ \in \mathcal{S}_T(\bG,\bP)$, satisfying the condition $\E^{\bQ}\left [\int_{0}^{T}(\alpha_s^{+})^2d\langle X \rangle_s\right]<\infty$.
\end{Theorem}
\begin{proof} 

Note that the condition $\E^{\bQ}\left [\int_{0}^{T}(\alpha_s^{+})^2d\langle X \rangle_s\right]<\infty$ implies that $\int \alpha^{+}dM$ is a true martingale. 

Let $\bQ \in \mathcal{S}_T(\bG,\bP)$. There exists a $(\bG,\bP)$ local martingale $L^\bQ$ such that $\frac{d\bQ}{d\bP}|_{\G_{t\wedge T}}=\mathcal E_t(L^\bQ)$, for $t\geq 0$. Further, by the Kunita-Watanabe decomposition,
\begin{equation} \label{LKW}
L^{\bQ}_t=\int_{0}^{t}\beta^{\bQ}_sdM_s +N^{\bQ}_t. 
\end{equation}
In the above, $N^{\bQ}$ is a $(\bG,\bP)$ local martingale orthogonal to $M$ and $\beta^{\bQ}$ is a $\bG$ predictable process. 
Under $(\bG,\bQ)$ and for $t\in[0,T]$, $X$ possesses the following decomposition:

\begin{equation*} 
X_t= M^\bQ_t+\int_{0}^{t}(\alpha_s+\beta^{\bQ}_s)d \langle M\rangle_s ,
\end{equation*}
 and  $M^{\bQ}=M-\langle L^{\bQ}, M \rangle 
= M -\int\beta^{\bQ} d\langle M \rangle$ is a $(\bG,\bQ)$ local martingale.

As $(X_{t\wedge T})$ is a $(\bG,\bQ)$ supermartingale,  for any $t\in [0,T]$ it must hold that
\be\label{alphabetaneg}
(\alpha_t+\beta_t^{\bQ})\ind_{d\langle M\rangle_t >0}\leq 0\quad \bP\; a.s.
\ee
We use the notation  $\{d\langle M\rangle_t> 0\}$, for the support of the random measure on $\mathcal B([0,T])$,
associated with the nondecreasing process $\langle M\rangle$.

We have that
\begin{align*}
H(\bQ|\bP) &=H(\bQ|\bQ^{\mathcal S})+\int \log R^{(+)}_T d\bQ\\
 &=H(\bQ|\bQ^{\mathcal S})+ \int \left(-\int_{0}^{T}\alpha^{+}_sdM_s-\frac{1}{2}\int_{0}^{T}(\alpha^{+}_s)^{2}d\langle M\rangle_s \right)d\bQ\\
&=H(\bQ|\bQ^{\mathcal S})+ \int \left(-\int_{0}^{T}\alpha^{+}_sdM^{\bQ}_s -\int_{0}^{T}\alpha^{+}_s\beta^{\bQ}_sd\langle M \rangle_s   -\frac{1}{2}\int_{0}^{T}(\alpha^{+}_s)^{2}d\langle M\rangle_s \right)d\bQ\\
& =H(\bQ|\bQ^{\mathcal S})+\E^{\bQ}\left[\int_{0}^{T}( -\alpha^{+}_s\beta^{\bQ}_s-\frac{1}{2}(\alpha^{+}_s)^{2})d\langle X\rangle_s\right ]. 
\end{align*}
Therefore,  
\begin{align*}
H(\bQ|\bP)-\frac{1}{2}\E^{\bQ}\left [\int_{0}^{T}(\alpha^{+}_s)^2d\langle X \rangle_s\right]=  H(\bQ|\bQ^{\mathcal S}) - \E^{\bQ}\left [\int_{0}^{T}\alpha^{+}_s(\alpha_s+\beta^{\bQ}_s)d\langle X\rangle_s\right ].
\end{align*}
and 
\begin{align*}
 H(\bQ|\bP)-\frac{1}{2}\E^{\bQ}\left [\int_{0}^{T}(\alpha^{+}_s)^2d\langle X \rangle_s\right]\geq \min_\bQ   H(\bQ|\bQ^{\mathcal S}) - \max_\bQ \E^{\bQ}\left [\int_{0}^{T}\alpha^{+}_s(\alpha_s+\beta^{\bQ}_s)d\langle X\rangle_s\right ],
\end{align*}
where the min and max above are taken over the same set as in the problem (\ref{minentropy}); it can be checked that $\bQ^{\mathcal S}$ belongs in this set. 

We have that  $\min_\bQ   H(\bQ|\bQ^{\mathcal S})=0$; the minimal value of $0$ is obtained if and only if $\bQ=\bQ^{\mathcal S}.$
Also, in view of the inequality (\ref{alphabetaneg}), $\E^{\bQ}\left [\int_{0}^{T}\alpha^{+}_s(\alpha_s+\beta^{\bQ}_s)d\langle X\rangle_s\right ]\leq 0$ and we can check that the maximum value of this expected value is 0, obtained when taking $\beta^{\bQ}=-\alpha^+$, that is,  $G^{\bQ}=R^{(+)}$, or $\bQ=\bQ^{\mathcal S}.$ This proves that (\ref{minentropy})  equals  $0$ and this value is attained if and only if $\bQ=\bQ^{\mathcal S}.$
\end{proof} 

Let us now emphasise some first connections between $\bQ^{\mathcal M}$ and $\bQ^{\mathcal S}$:

\begin{Theorem}\label{TheoremGeneral}We assume that the representation in Hypothesis \ref{H2} holds. 
\begin{itemize}
\item[(a)]If the minimal martingale measure $\bQ^{\mathcal M}$ exists then,  $\bQ^{\mathcal S}\in \mathcal{S}(\bG,\bP)$. 
 \item[(b)] Suppose that the process $(\ind_{\alpha_t> 0})_{t\geq 0}$ admits a right-continuous modification having bounded total variation. 
Under this condition,  if $\bQ^{\mathcal M}\in \mathcal M(\bG,\bP)$ then $\bQ^{\mathcal S}\in \mathcal{S}(\bG,\bP)$.
 \end{itemize}
\end{Theorem}

\brem
Before we prove this result, let us illustrate the condition appearing in (b). \\ We consider $W$ is a Brownian motion with respect to $(\bG,\bP)$, and let $\alpha:=|W|$. We obtain $\ind_{|W_t|>0}=H_t-\ind_{W_t=0}$ where $H_t:=\ind_{|W_t|\geq 0}=1$. By Proposition 3.12 page 108 in \cite{RevuzYor},  the set $\{t: W(t)=0\}$ is a.s. closed, without any isolated points and has zero Lebesgue measure. Hence, $H$ is a modification of the process $\ind_{|W_t|>0}$. As $H$ is continuous, with null total variation, we conclude that for this example, the condition stated in Theorem \ref{TheoremGeneral} (b) is fulfilled.  \\
Alternatively, let us consider there is a  $(\bG,\bP)$ counting process $N$ with c\`adl\`ag sample path and with finite explosion time. Let $\alpha:=(-1)^{N}$. Then, the process $\alpha$ is c\`adl\`ag and has infinite total variation. In this case, the  condition stated in Theorem \ref{TheoremGeneral} (b) is not fulfilled.
\erem

\proof (of Theorem \ref{TheoremGeneral})
\begin{itemize}

\item[(a)]We need to prove that if $R$ is a strictly positive and square integrable $(\bG,\bP)$ martingale then $R^{(+)}$ is a strictly positive and uniformly integrable $(\bG,\bP)$ martingale.  Let us introduce the following $(\bG,\bP)$ local martingale, which is orthogonal to  $R^{(+)}$:
\[
R^{(-)}_{t}=\mathcal E_t\left(\int_0^\cdot (\alpha_u)^-d M_u\right),\quad t\geq 0,
\]
so that:
\[
R_t:=R^{(+)}_{t}R^{(-)}_{t}.
\]
The processes $R^{(+)}$ and $R^{(-)}$ are strictly positive local martingales: they satisfy $R_0^{(+)}=R_0^{(-)}=1$ and can become negative or null only after a jump, i.e., one at a time at most (as they do not have common jumps). In other words, the product $R$ is strictly positive if and only if $R^{(+)}$ and $R^{(-)}$ are strictly positive. Then, one can show that  $R^{(+)}$ is a uniformly integrable martingale, as a direct application of Proposition \ref{Y=UZ} in Appendix \ref{appendix}.

\item[(b)] 
  If $ \bQ^{\mathcal M}\in \mathcal M(\bG,\bP)$, then  $R=\mathcal E(L)$ is a strictly positive, uniformly integrable $(\bG,\bP)$ martingale, where $L=-\int \alpha dM$. We denote by $H=(H_t)_{t\geq 0}$ a right-continuous modification of the process  $(\ind_{\alpha_t>0})_{t\geq 0}$, assumed of bounded variation. We can write $R^{(+)}_t=\mathcal E_t(\int HdL)$ and $R^{(-)}_t=\mathcal E_t(\int (1-H)dL)$.   The result then follows as an application of the Lemma \ref{M=UZ(2)} in Appendix \ref{appendix}. 
\end{itemize}
\finproof

Another sufficient condition for $\bQ^{\mathcal S}\in \mathcal{S}(\bG,\bP)$ is the following:

\begin{Lemma}We assume that the representation in Hypothesis \ref{H2} holds. 

Suppose that it exists $ \bQ\in  \mathcal{S}_T(\bG,\bP)$ such that  $d\bQ/d\bP|_{\G_t}=\mathcal E_t(L^\bQ)$  is satisfying $\E^{\bP}[e^{\frac{1}{2}\langle L^\bQ\rangle_T }]<\infty$. Then, $\bQ^{\mathcal S}\in\mathcal{S}_T(\bG,\bP)$.
\end{Lemma}

\proof 
We  use the same notation as in the proof of Theorem \ref{Entropy}.; in particular the process $L^\bQ$ is a $(\bG,\bP)$ local martingale with the Kunita Watanabe decomposition given in (\ref{LKW}).
  In order for $X$ to be a $(\bG,\bQ)$-supermartingale, necessarily  the process:
\[
 \int_0^t \left(\alpha_u+\beta^\bQ_u\right)d \langle M\rangle _u,\quad t\in[0,T]
\]
is nonincreasing, in particular we need:
\[
0\leq \alpha \ind_{\{\alpha>0\}}\leq - \beta^\bQ \ind_{\{\alpha>0\}}\quad \bP\;a.s.
\](notice that $\{(\omega,t)|\alpha_t(\omega)>0\}$ is supposed to be included in the support of $d \langle M\rangle$ by Hypothesis \ref{H2}). 
By assumption, $\mathcal E (L^\bQ)$ satisfies the Novikov condition. Hence, in view of the above inequality, $R^{(+)}$ is a uniformly integrable martingale satisfying the Novikov condition.
\finproof

\section{The (H) hypothesis as a no arbitrage condition for insiders in presence of short sales restrictions}\label{secH}

In this section, assuming $\bF$-(PRP), we further analyse connections between the existence of $\bQ^{\mathcal M}$ and $\bQ^{\mathcal S}$, in particular under the assumption that the (H) hypothesis under some probability measure on $(\Omega, \G)$ holds true. Given two filtrations, $\bF\subset \bG$, the (H) hypothesis, also called the immersion property, is the property that all $\bF$ martingales are also $\bG$ martingales; the immersion property holds in relation to a reference probability measure. In the financial literature, the (H) hypothesis under an equivalent martingale measure for the asset prices is known as a classical no arbitrage condition whenever asset prices are $\bF$ adapted, but agents may employ $\bG$-adapted strategies. The main result is due to Blanchet-Scalliet and Jeanblanc \cite{bsjeanblanc} (see Theorem \ref{HNA} below),  further results involving the (H) hypothesis and no arbitrage can be found in  \cite{CocJeanNik12}. It is to be noted that both references  \cite{bsjeanblanc} and  \cite{CocJeanNik12} are focused on the particular case where the filtration $\bG$ is obtained from $\bF$  via a progressive enlargement with a random time, but some of their results are general (as we shall emphasize below).  These results are investigating the link between (NFLVR) in $\bG$ and (H), more precisely, when  the $\bG$ informed agent is able to sell short.  We shall see that  the immersion property is still a key condition for understanding (NFLVRS) in $\bG$. 

For the reader's convenience, let us first give a definition and recall some results that are needed later on.
\begin{Definition}
We say that the immersion property, or (H) hypothesis, holds between $\bF$ and $\bG$ under the probability measure $\bQ$ if  all $(\bF,\bQ)$ local martingales are $(\bG,\bQ)$ local martingales. When this property holds, we write $\bF\overset{\bQ}{\hookrightarrow}\bG$. When this property does not hold, we write  $\bF\overset{\bQ}{\nhookrightarrow}\bG$.
\end{Definition}

More about the immersion of filtrations can be found in \cite{bremaudyor}, \cite{DelMeyer1} or \cite{aj}. See also \cite{Kusuoka}, \cite{CocJeanNik12} where changes of the probability measure are considered together with the immersion property.  

In general, the immersion property is not preserved by changes of the probability measure.  Let us introduce the following notation for the set of probability measures that have the immersion property and are equivalent to $\bP$:
\begin{equation*}
 \mathcal I(\bF,\bG,\bP):=\left \{ \widetilde \bQ \text{ probability on }(\Omega,\G)\; |\;
  \bP\sim \widetilde \bQ,\;\bF \overset{\widetilde \bQ}{\hookrightarrow}\bG
 \right  \}.
\end{equation*}

We now state the result of Blanchet-Scalliet and Jeanblanc that was mentioned above. Even though in \cite{bsjeanblanc}  the filtration $\bG$ is defined as a progressive enlargement of $\bF$ with a random time, the proof of their Proposition 1 -  that corresponds to Theorem \ref{HNA} below - holds actually in the generality of our definition for $\bG$, as one can easily verify; see also Corollary 4.6 in \cite{CocJeanNik12} for an alternative proof of this result.

\begin{Theorem}\label{HNA} We suppose that $\bF$-(PRP) holds and $\mathcal M(\bG,\bP)\neq \emptyset$. Then, $\mathcal M(\bG,\bP)\subset  \mathcal I(\bF,\bG,\bP)$.

\end{Theorem}

We want to understand if a potential of arbitrage is generated by the inside information, and the possible connections with the immersion property. To fully isolate this case of interest, no arbitrage for the common investors will be systematically assumed from now, in the form of  Hypothesis \ref{H1}. This way, any possible $\bG$ arbitrage is necessarily arising from the exploitation of the inside information and is not accessible with $\bF$ adapted strategies.  With this extra requirement, we can be establish some stronger result than Theorem \ref{HNA}, in the form of  Theorem \ref{ThmGenH} below. 

But before, let us also recall some well known financial implication of $\bF$-(PRP), namely the market is complete for the common investors. 
\begin{rem}\label{remPstar}
Suppose that Hypothesis \ref{H1} and $\bF$-(PRP) hold. Then, there exists a unique probability $\bP^*$ on $(\Omega,\F_\infty)$ such that:
\begin{itemize}
\item[(i)] $\bP^*\sim\bP$ on $\F_\infty$, and 
\item[(ii)]$X$ is a $(\bF,\bP^*)$ local martingale. 
\end{itemize}
This result is known under the name of  second theorem theorem of asset pricing (see Harrison and Pliska \cite{HarrPlisk83}).
Furthermore,  $X$ has the  $(\bF,\bP^*)$ predictable representation property in the filtration $\bF$. Indeed,  the predictable representation property is stable under equivalent changes of the probability measure (Theorem 13.12 in \cite{HeWangYan92}). \end{rem}
One question is whether there are extensions of $\bP^*$ to the larger sigma-field $\G$ that can serve as equivalent martingale measures for the insider. 
\begin{Theorem}\label{ThmGenH}
We suppose that Hypothesis \ref{H1} and $\bF$-(PRP) hold. Then, the following are equivalent:
\begin{itemize}
\item[(i)] $\mathcal M(\bG,\bP)\neq \emptyset$ 
\item[(ii)] $\mathcal I(\bF,\bG,\bP)\neq \emptyset$
\item[(iii)]  $ \mathcal I(\bF,\bG,\bP)\cap \mathcal M(\bF,\bP)= \mathcal I(\bF,\bG,\bP)\cap \mathcal M(\bG,\bP)\neq \emptyset$. 
\item[(iv)] $ \mathcal I(\bF,\bG,\bP)\cap  \mathcal{S}(\bG,\bP)\neq \emptyset$. 
\end{itemize}
\end{Theorem}
\proof
\begin{itemize}
\item[-] The implication (i) $\Rightarrow$ (ii) follows from Theorem \ref{HNA}. 

\item[-] The  implications (ii) $\Rightarrow$ (iii) follows from Proposition 4.4 in  \cite{CocJeanNik12}. We emphasize that this result holds in the generality of our setting for the filtrations $\bF$ and $\bG$. We recall this result, for the reader's convenience:  \\
\textit{Given a reference probability measure $\bP$, if  $\mathcal I(\bF,\bG,\bP)\neq \emptyset$, then there is  $ \bQ\in\mathcal I(\bF,\bG,\bP)$  such that every $(\bF,\bP)$ martingale is a $(\bG, \bQ)$ martingale, that is,  $\bQ=\bP$ on $\F_\infty$. }\\
This result means that if (ii) holds, then there exist equivalent probability measures that  keep the $\bF$ martingales unchanged and at the same time  ensure the immersion property holds. We now take $\bP^*$ as a reference measure, where $\bP^*\in\mathcal M(\bF,\bP)$, that is, $X$ is a local martingale in the smaller filtration $\bF$ under $\bP^*$. Because $\bP\sim \bP^*$, it follows that $\mathcal I(\bF,\bG,\bP^*)=\mathcal I(\bF,\bG,\bP)$. Thence (ii) is equivalent to: $\mathcal I(\bF,\bG,\bP^*)\neq \emptyset$.
Then, by Proposition 4.4 in  \cite{CocJeanNik12} stated above, there is $ \bQ\in\mathcal I(\bF,\bG,\bP^*)$  such that $\bQ=\bP^*$ on $\F_\infty$, in particular  $X$ is a local martingale in both filtrations under $\bQ$. Hence (iii)   holds.  

\item[-] The implication (iii) $\Rightarrow$ (i) is obvious.

\item[-]The equivalences with the last statement are also obvious: (iv) implies (ii) and also (iii) implies (iv), as if $\bQ$ is a local martingale measure in filtration $\bG$,  it is also a supermartingale measure in filtration $\bG$.
\end{itemize}
\finproof

A few comments on the results in Theorem \ref{ThmGenH} are as follows. For the insider, as soon as there exists a probability measure equivalent to $\bP$ under which the (H) hypothesis holds, there are no arbitrages to undertake,  without or even involving short positions: $\emptyset\neq \mathcal M(\bG,\bP)\subset  \mathcal{S}(\bG,\bP)$. But, on the other hand, if there is not such an equivalent probability measure, then there are automatically arbitrages, possibly involving short-selling, as $\mathcal M(\bG,\bP)= \emptyset$.

 An interesting case is the one where there are arbitrage opportunities for the insider, but they all require short-selling by the insider.  This  corresponds to the case $\mathcal M(\bG,\bP)= \emptyset$ and $\mathcal{S}(\bG,\bP)\neq  \emptyset$.   Theorem \ref{ThmGenH}  tells us that such a situation corresponds to the case where no element of $\mathcal{S}(\bG,\bP)$ can ensure for the (H) hypothesis to hold, when taken as a reference measure. Indeed, another formulation for the equivalence between (i) and (iv) is:
 \begin{Cor} We suppose that Hypothesis \ref{H1} and $\bF$-(PRP) hold. Under these assumptions, 
 $\mathcal M(\bG,\bP)=\emptyset$ if and only if $\mathcal S(\bG,\bP)\cap  \mathcal I(\bF,\bG,\bP)=\emptyset$.
\end{Cor}

We illustrate this situation by a simple example using honest times, inspired by the more general construction available in \cite{Nikegh06a}. We refer to \cite{FontJeanSong14} for an analysis of arbitrages arising from knowledge of honest times (where short-sales restrictions are not studied). These references provide a general setting, surely also suitable for a more systematic study of the implications of short-sale restrictions in progressive enlargements of filtrations with honest times (that is beyond the scope here). 

\begin{example*}We assume the filtration $\bG$ is the progressive enlargement of the filtration $\bF$ with the last passage time of the stock price at its supremum level. More exactly let $X=\mathcal E(\sigma W_t)$ where $W$ is $(\bF,\bP)$ Brownian motion. We notice that  $\lim_{t\to\infty } X_t =
0$. Let $S_t = \sup_{s\leq t} X_s$ and
\[
g = \sup\{t\geq 0: X_t=S_\infty\} = \sup\{t\geq 0: S_t-X_t=0\}.
\]We let $\bG$ be the smallest filtration that satisfies the usual assumptions and  so that $g$ is a $\bG$ stopping time. We thus assume the insider observes $g$ when it occurs. 
The random time $g$ is a last passage time within the filtration $\bF$ and moreover the decomposition of the price process in the filtration $\bG$ is given by (cf. Proposition 2.5 in \cite{NikeghYor06}):
\begin{align*}
X_t&=1+\int_0^{t\wedge g} \frac{1}{X_s}d\langle X\rangle_s -\int_{g}^{t\vee g} \frac{1}{S_s-X_s}d\langle X\rangle_s+M_t\\
&=1+\sigma^2\int_0^{t\wedge g}X_sds -\sigma^2\int_{g}^{t\vee g} \frac{X^2_s }{S_s-X_s}ds+\sigma \int_0^tX_sd\hat W_s.
\end{align*}
Where $M= \sigma \int Xd\hat W$ corresponds to the local martingale in the decomposition (\ref{Xdecomp}) with  $\hat W$ being a $(\bG,\bP)$ Brownian motion. We now consider a finite horizon, that is $t\in[0,T]$, so that the process $X$ is stopped at $T$. Then,  there exists a minimal supermartingale measure and it satisfies $d\bQ^{\mathcal S}= \mathcal E (- \sigma \hat W_{T\wedge g})\cdot d\bP$ on $\G_T$. Indeed, $\bQ^{\mathcal S}\in\mathcal S_T(\bG,\bP)$ as one can easily verify.  However, there is no martingale measure in the filtration $\bG$: $\mathcal I_T(\bG,\bP)=\emptyset$ and $\mathcal M_T(\bG,\bP)=\emptyset$. Indeed, progressive enlargement with honest times are well known examples of models where insiders can have free lunches with vanishing risk (see for instance Theorem 4.1 in \cite{FontJeanSong14}). In conclusion, in this example the insider needs short-selling the stock in order to realise an arbitrage. Indeed, selling the asset $X$ short at $g\wedge T$ and holding the cash obtained, then unwind the positions at  $T$ provides an arbitrage (NFLVR). In conclusion, short-sales prohibitions can be effective in eliminating arbitrages for this example of insider information.

\end{example*}

Now, we want to shed some additional light on the connections between $\bQ^{\mathcal S}$ and $\bQ^{\mathcal M}$, namely, we provide conditions under which both fail to be probability measures (in Theorem \ref{ThmH} below).
As we shall only focus on a bounded time horizon $[0,T]$, we introduce the stopped filtrations $\bF_T=(\F_{T\wedge t})$ and $\bG_T=(\G_{T\wedge t})$, and we suppose that all stochastic processes are also stopped at $T$ without introducing further notation.  We also shall assume that $\bF_T\overset{\bP}{\nhookrightarrow}\bG_T$, that is, so that our model is compatible with the (possible) failure of  (NFLVR) in the filtration $\bG$ (the terminal value of the portfolio will be nonnegative and strictly positive with positive probability).

 Let us introduce the following set of probability measures:
\begin{equation*}
 \mathcal P_T(\bF,\bG,\bP):=\left \{ \widetilde \bQ \text{ probability on }(\Omega,\G_T)\;\biggm | \;
 \begin{array}{l}
 \bF_T\overset{\widetilde \bQ}{\hookrightarrow}\bG_T,\\
\bP\ll\widetilde \bQ;\; \bP\sim \widetilde \bQ \text{  on }\F_T.
 \end{array}
 \right  \}.
\end{equation*}
We have that $ \mathcal I(\bF,\bG,\bP)\subset  \mathcal P_\infty(\bF,\bG,\bP)$.

Absolutely continuous, but not equivalent, changes of measure may lead to arbitrage opportunities in the (NFLVR) sense, see e.g. Delbaen and Schachermayer \cite{DelbScha95}, Osterrieder and Rheinl\"ander \cite{OstRhe06},  Ruf and Runggaldier \cite{RufRung13} and Chau and Tankov \cite{ChauTank13}.

\begin{notation}
Suppose that $\bQ\in \mathcal P_T(\bF,\bG,\bP)$.  We denote by $V_T(\bQ)$ the set of $(\bG_T,\bQ)$  local martingales that admit the representation $(c+\int_0^{\cdot\wedge T} H_s dm_s)$, where $c$ is constant, $m$ is an $(\bF_T,\bQ)$  local martingale  and $H$ is a $\bG_T$ predictable process such that $\int_0^{\cdot\wedge T} H_s^2d\langle m\rangle_s $ is locally integrable. As the immersion property holds for any $\bQ\in \mathcal P_T(\bF,\bG,\bP)$, the stochastic integrals introduced above are well defined.

\end{notation}

In our next two results we will show that for $\bQ^{\mathcal M}$ or $\bQ^{\mathcal S}$ to be probabilities on $\G_T$, a certain subset of $ \mathcal P_T(\bF,\bG,\bP)$ plays a central role,  namely
\begin{equation*}
\mathcal p_T(\bF,\bG,\bP):=\left\{\widetilde \bQ\in \mathcal P_T(\bF,\bG,\bP)
\; \Big| \;
  \frac{d\bP}{d\widetilde\bQ }\in V_T(\widetilde\bQ) 
 \right  \}.
\end{equation*}

\begin{Theorem}\label{ThmQstar}
We assume that Hypothesis \ref{H1} and $\bF$-(PRP) hold. If $p_T(\bF,\bG,\bP)\neq \emptyset$, then there exists a  unique probability measure $\bQ^*$ that satisfies:
\begin{itemize}
\item[(i)] $\bQ^*\in \mathcal p_T(\bF,\bG,\bP)$.
\item[(ii)] The price process $X$ is a $(\bG_T, \bQ^*)$ local martingale.
\end{itemize}
Assume furthermore  Hypothesis \ref{H2}  and assumption (C) hold. Then:
\begin{itemize}
\item[(iii)]  The process  
$$
D^*_t:=\E^{\bQ^*}\left [\frac{d\bP}{d\bQ^{*}}\;\Big | \;\G_{t\wedge T}\right],\quad t\geq 0
$$ admits the $(\bG_T,\bQ^*)$ representation:
\begin{equation}\label{defD}
D^*_t=1+\int_0^t D^*_{s}\alpha_sdX_s.
\end{equation}
It follows that $\bQ^*\sim \bP$  on $\G_T$ if and only if $\bQ^*= \bQ^\mathcal M\in\mathcal M_T(\bG,\bP)$.
\end{itemize}
\end{Theorem}

\proof 
The set $\mathcal p_T(\bF,\bG,\bP)$ contains equivalent probabilities and therefore we shall introduce $\bQ^*$ via a Girsanov transformation,  relatively to an arbitrary element in $\mathcal p_T(\bF,\bG,\bP)$.

Let $ \bQ\in \mathcal p_T(\bF,\bG,\bP)$. 
Then there is a nonnegative $(\bG_T, \bQ)$ martingale $E$ with $d\bP=E_T\cdot  \bQ$ on $\G_T$ and a strictly positive $(\bF_T, \bQ)$ martingale $e$ with $d\bP=e_T\cdot  \bQ$ on $\F_T$.  

If Hypothesis \ref{H1} together with $\bF$-(PRP) hold, there is a unique local martingale measure for the price process $X$ and in the filtration $\bF$, that is, there is a unique  $(\bF, \bP)$ uniformly integrable martingale $\delta$ with $\delta$ being strictly positive $\delta_0=1$, such that the process $X\delta$ is  $(\bF, \bP)$ local martingale.  Furthermore,  $\delta$ admits a representation as $$\delta_t=\mathcal E_t\left(-\int_0^\cdot a_s dm_s^{\bF,\bP}\right)\quad t\geq 0$$ (in line with the representation of the asset price given in (\ref{XdecompF})). 
Hence, an equivalent local martingale measure for the common ($\bF$ informed) investors may be defined on $\G$ via
$$
d\bP^*:=\delta_\infty \cdot d\bP\quad \text{ an }\G.$$

We are after a probability $\bQ^*\in  \mathcal p_T(\bF,\bG,\bP)$ that may potentially serve as an equivalent martingale measure for the insider, as the price process $X$ is required  to be a $(\bG_T, \bQ^*)$ local martingale, by statement (ii) of the theorem. We introduce $\bQ^*$ via:  
$$
d\bQ^*=\delta_T e_T\cdot  d\bQ\quad \text{ on $\G_T$}.
$$  
It can be checked that $\bQ^*\sim  \bQ$ and furthermore $\bQ^*$ concides with $\bP^*$ on $\F_T$:
\[
\frac{d\bQ^*}{d  \bQ}\Big|_{\F_T}=\delta_t \times e_t=\frac{d\bP^*}{d \bP}\Big|_{\F_T}\times\frac{d\bP}{d  \bQ}\Big|_{\F_T}=\frac{d\bP^*}{d  \bQ}\Big|_{\F_T}.
\]
Consequently,  $\bQ^*$ is a well defined change of measure on $\bF_T$ with the Radon-Nikod\'ym density process $(\delta_t e_t)$ being a strictly positive $(\bF,  \bQ)$ martingale. Finally, as the (H) hypothesis holds under $ \bQ$, the process $(\delta_t e_t)$ is also a  $(\bG,  \bQ)$ martingale, so that the mesure $\bQ^*$ is well defined also on $\G_T$. 

One can see that $\bQ^*\in  \mathcal P_T(\bF,\bG,\bP)$. This is the consequence of the fact that the immersion property is preserved under changes of measure using $\F_\infty$  measurable Radon-Nicod\'ym derivatives (see Lemma 4.3 in\cite{CocJeanNik12}).  Indeed, we have $\bF \overset{ \bQ}{\hookrightarrow}\bG$ and $\frac{d\bQ^*}{d \bQ}$ is $\F_T$ measurable therefore $\bF \overset{ \bQ^*}{\hookrightarrow}\bG$ and so that indeed 
\begin{equation}\label{QstarP}
\bQ^*\in  \mathcal P_T(\bF,\bG,\bP).
\end{equation} 
We can now prove the statements of the theorem:
\begin{itemize}
\item[(ii)] Using Girsanov's theorem in the filtration $\bF_T$, we obtain that the following is a $(\bF_T,\bQ^*)$ martingale:

\[
m_t^{\bF, \bQ^*}: = m_t^{\bF, \bP}-\int_0^t \frac{d\llangle\delta, m^{\bF, \bP}\rrangle_s}{\delta_{s-}}= \int_0^t a_s d\llangle m^{\bF, \bP} \rrangle_s + m_t^{\bF, \bP}=X_t-X_0.
\]As $\bF \overset{ \bQ^*}{\hookrightarrow}\bG$, it follows that $m^{\bF, \bQ^*}$ is also a $(\bG_T,\bQ^*)$ martingale, so that the $(\bG_T, \bQ^*)$ semimartingale decomposition of the asset price $X$ is:
\begin{equation}\label{intqpG}
X_t=X_0+m^{\bF,  \bQ^*}_t.
\end{equation}
 that is a $(\bG_T,\bQ^*)$ martingale.

\item[(i)] In view of (\ref{QstarP}), we only need to show that $\frac{d\bP}{d \bQ^*}|_{\G_T}=\frac{d\bP}{d  \bQ}|_{\G_T}\times \frac{d  \bQ}{d \bQ^*}|_{\G_T}=\frac{E_T}{e_T\delta_T}\in V_T(\bQ^*)$. We first provide $(\bG_T, \bQ)$ representations for the processes $e$ and $E$. 

Because $\bF$-(PRP) holds and $ \bQ\sim\bP$ on $\F_T$, it follows that in the filtration $\bF$ the $ \bQ$ martingales also have the predictable representation property, with respect to the  martingale $ m^{\bF, \bQ}= m^{\bF, \bP}-\int\frac{d\llangle e_{-},m^{\bF, \bP}\rrangle}{e_{-}}$ (see \cite{HeWangYan92}, Theorem 13.12]). Then, there exists an $\bF$ predictable process $h$ so that 
$$
e=\mathcal E\left (\int hdm^{\bF, \bQ}\right).
$$ 
As (H) holds under $ \bQ$, $m^{\bF, \bQ}$ is also a $ \bQ$ martingale in the larger filtration $\bG_T$. Then,   $e$ is also a  $(\bG_T, \bQ)$ martingale,  while  $E$ has the representation 
$$
E=\mathcal E\left(N+\int H dm^{\bF, \bQ}\right)
$$ with $N$ a $(\bG_T, \bQ)$ local martingale orthogonal to $m^{\bF, \bQ}$, and where $H$ is a $\bG$ predictable process. This representation of $E$ is $ \bQ-a.s.$ unique on the set $[0,T^0]$ with $T^0=\inf\{t\geq 0: E_t=0\}$. 
As $ \bQ\in \mathcal p_T(\bF,\bG,\bP)$, we observe that $N_{\cdot \wedge T_0}\equiv 0$ since  $E\in  V_T( \bQ)$. This leads to the representation  
$$
E=\mathcal E\left (\int H dm^{\bF, \bQ}\right).
$$
Moreover, the condition $\E_{ \bQ}[E_t|\F_t]=e_t$ leads to a relation between $H$ and $h$, (cf.  Br\'emaud and Yor \cite{bremaudyor}, Proposition 1) namely
$$
(h_t)\text{ is the $(\bF,\bP)$ predictable  projection of }(H_t).
$$ 

 Using all the above representations, we obtain:
\[
D^*_T:=\frac{d\bP}{d \bQ^*}|_{\G_T}=E_T(e_T\delta_T)^{-1}=\mathcal E_T\left (\int_0^\cdot H_t-(h_t+a_t)dX_t\right)\in V_T(\bQ^*)
\]
that proves the statement (i). 
\end{itemize}
We now assume Hypothesis \ref{H2} and (C) and prove the last statements.
\begin{itemize}
 \item[(iii)]  It is sufficient to prove the relation $R=(D^*)^{-1}$ $\bP$ a.s. Let us denote $\beta:= (h-H+a)$ so that $D^*_t:\frac{d\bP}{d \bQ^*}|_{\G_t}=\mathcal E_t\left (\int_0^\cdot \beta_sdX_s\right)$. Given   $ \bQ^*\ll\bP$ on $\G_T$,  we may apply a Girsanov transformation and obtain that $\hat M_t:=m^{\bF,\bQ^*}_t-\int_0^t\beta_s d\langle m^{\bF,\bQ^*}\rangle_t$ is a $(\bG, \bP)$  local martingale. Given the assumption  (C),  $\langle m^{\bF,\bQ^*}\rangle=\langle \hat M\rangle$, that is, the quadratic variation.  Consequently, we obtain the following $(\bG, \bP)$ representation of  $X$:
\begin{equation}
X_t=X_0+\int_0^t\beta_s d\langle \hat M\rangle_t+\hat M_t.
\end{equation}
By analysing Hypothesis \ref{H2}, we see that the uniqueness of a Doob Meyer decomposition imposes the following relations $\bP$ a.e.: $\hat M=M$, leading to $\langle \hat M\rangle =\langle M\rangle=\langle X\rangle$ and also $\beta\ind_{\{d\langle M\rangle >0\}} =\alpha\ind_{\{d\langle M\rangle >0\}}$. This proves that indeed $D^*_t=\mathcal E_t\left (\int_0^\cdot \alpha_tdX_t\right)=\mathcal E_{t\wedge T^0}\left (\int_0^\cdot \alpha_tdX_t\right)$ holds $\bP$-a.s.. The equality also holds  $\bQ$-a.s  because the possibly additional null sets of $\bP$ are included in $\{D_T^*=0\}=\{E_T=0\}=\{T^0<T\}$.
It follows that indeed $R=(D^*)^{-1}$ $\bP$ a.s. and is a true $(\bG,\bP)$ martingale, if and only if $\bQ^*(T^0<T)=0$, that is $\bQ^*\sim\bP$.
\end{itemize}

 \finproof

If  $\bP$ and $\bQ^*$ are not equivalent on $\G_T$, it may still be the case that $\mathcal M_T(\bG,\bP)\neq \emptyset$. Indeed, one cannot exclude a priori the existence of a probability $\widetilde \bQ$ satisfying $\widetilde \bQ\in\mathcal I(\bF,\bG,\bP)$ ( by Theorem \ref{ThmGenH} (ii) this is indeed equivalent to $\mathcal M_T(\bG,\bP)\neq \emptyset$).  In such a case the measure $\bQ^m$ is not a probability, but there exist equivalent local martingale measures for the price process in $\bG$.

In the next theorem we emphasise that quite strong relations between $\bQ^\mathcal M$ and $\bQ^\mathcal S$ exist whenever the reference measure $\bP$ is a local martingale measure for the price process $X$, within the smaller filtration $\bF$ (i.e., this situation is obtained for $(a_t)\equiv 0$ in (\ref{XdecompF}), or,  equivalently $\bP=\bP^*$ on $\F_\infty$, where $\bP^*$ is as in Remark \ref{remPstar}).
\begin{Theorem}\label{ThmH}
 We assume Hypotheses \ref{H1} and \ref{H2} together with conditions (C) and $\bF$-(PRP).  In addition, we consider $X$ is a $(\bF,\bP)$ local martingale (in the representation (\ref{XdecompF}), the process $(a_t)$ is taken to be null). Suppose  $\mathcal p_T(\bF,\bG,\bP)\neq \emptyset$ and let $\bQ^*$ be such as in Theorem \ref{ThmQstar}.
 Then, the following hold:
\begin{itemize}
\item[(a)] If   $\bP$ and $\bQ^*$ are not equivalent on $\G_T$,    then $\bQ^{\mathcal M}\notin \mathcal M_T(\bG,\bP)$ and $\bQ^{\mathcal S}\notin \mathcal{S}_T(\bG,\bP)$. \\
\item[(b)] If $\bP\sim\bQ^*$  on $\G_T$  and $\frac{d\bP}{d\bQ^*}|_{ \G_T}\in L^2(\Omega,\G,\bQ^*)$, then $\bQ^{\mathcal S}\in  \mathcal{S}_T(\bG,\bP)$ and $\bQ^{\mathcal M}\in\mathcal M_T(\bG,\bP)$.
\end{itemize}
\end{Theorem}

The proof of this theorem is postponed at the end of this section. 

\bigskip

In order to prove Theorem \ref{ThmH},  we take as  reference a probability space $(\Omega, \G, \bQ^*)$, endowed with right-continuous filtrations $\bF^0\subset \bG^0$ satisfying (\ref{fg0}). We suppose that (H) hypothesis holds under $\bQ^*$ and our aim is to recover $\bP$ form $\bQ^*$ with the help of a Girsanov transformation. The key properties of our initial model under $\bP$, in particular  $R$ and $R^{(+)}$  being or not strict local martingales, will appear as resulting from whether $\bP$ is actually equivalent to $\bQ^*$ or only absolutely continuous (see Proposition \ref{ReprXfromQ} below).  The filtrations $\bF$ and $\bG$ will be defined as  completions with the null sets of $\bP$, as before.

The assumptions that are going to be used below are the following:
\begin{itemize}
 \item[(A1)] The filtration $ \bF^0$ is the natural, $\bQ^*$ augmented filtration of a $\bQ^*$-Brownian motion $B$. 
 \item[(A2)] The filtrations $\bF^0$ and $\bG^0$ satisfy the condition (\ref{fg0})  and furthermore $ \bF^0\overset{\bQ^*}{\hookrightarrow} \bG^0$.
\item[(A3)] There is a nonnegative $(\bG^0,\bQ^*)$ martingale having the representation:
\begin{equation}\label{defD}
D^*_t=1+\int_0^t G_sdB_s,
\end{equation}
where $G$ is a $\bG^0$ predictable process satisfying  $\E_{\bQ^*}[G_t|\F^0_t]=0$ for all $t\in[0,T]$. 
\end{itemize}
We now define the probability $\bP$ as:
 \begin{equation}\label{dPdQ}
 \frac{d\bP}{d\bQ^*}\Big|_{\G^0_T}=D^*_{T},
 \end{equation}
where $D^*$ is the martingale defined in (\ref{defD}).

 \begin{Proposition}\label{ReprXfromQ} We consider the probability space $(\Omega,\G,\bQ^*)$ with the assumptions (A1)-(A3) and let $\bP$ be defined as in (\ref{dPdQ}). The filtration $\bF$ is defined as the usual augmentation of  $\bF^0$ with the $\bP$ null sets and similarly, the filtration $\bG$ is defined as the usual augmentation of  $\bG^0$ with the  $\bP$ null sets. 
 
 Let $X$ be  a $(\bF,\bQ^*)$ local martingale.  The following hold:
 \begin{itemize}
 \item[(a)] The  process   $X$ is a $(\bF,\bP)$ local martingale and it is a $(\bG,\bP)$ semimartingale with the Doob Meyer decomposition:
 \begin{equation}\label{Xrepr}
 X_t=X_0+M_t + \int_0^t \alpha_s d\langle M\rangle_s\quad t\in[0,T]
 \end{equation}
 for a $(\bG,\bP)$ local martingale $M$ and a $\bG$-predictable process $\alpha$ having a null $(\bF,\bP)$ optional projection. Thus, $X$ satisfies the Hypothesis \ref{H2} from Section \ref{SectionSetting} under $\bP$, for $t\in[0,T]$.
 \end{itemize}
 Define $R:= \mathcal E\left(-\int_0^\cdot \alpha_ud M_u\right)$ and $R^{(+)}:= \mathcal E\left(-\int_0^\cdot \alpha^+_ud M_u\right)$.
\begin{itemize}
 \item[(b)] If $\bP\sim\bQ^*$ on $G_T$ and $\E_{\bQ^*}[(D^*_T)^2]<\infty$, then, both $R$ and $R^{(+)}$ are true  $(\bG_T,\bP)$-martingales. 
 \item[(c)]  If $\bP\ll\bQ^*$  but $\bP$ is not equivalent to $\bQ^*$ on $G_T$, then both $R$ and $R^{(+)}$ are strict  $(\bG_T,\bP)$-local martingales.
 \end{itemize}
 \end{Proposition}

 \proof For simplicity, all processes are considered as stopped at $T$.\\
\textit{Proof of claim in (a).}  The process $X$ being an  $(\bF,\bQ^*)$ local martingale by assumption, there is an $\bF$-predictable process $F$  such that
\begin{equation}\label{XFQ}
X=X_0+\int F dB.
\end{equation}
We shall prove that the assumption (A4) implies:
  \begin{equation}\label{LF}
 \frac{d\bP}{d\bQ^*}\Big|_{\F_T}=1
 \end{equation}
which then implies that all  $(\bF,\bQ^*)$-local martingales remain  $(\bF,\bP)$ local martingales, in particular $X$.

The density process $(\frac{d\bP}{d\bQ^*}|_{\F_t})_{t\in[0,T]}$ is given by the $(\bF,\bQ^*)$ optional projection of  the process $D^*$. The Brownian motion $B$ being an $(\bF,\bQ^*)$ and a $(\bG,\bQ^*)$ local martingale (by assumption (A2)), we obtain (see \cite{bremaudyor}):
 \[
\frac{d\bP}{d\bQ^*}|_{\F_t}=\;^{o,\bQ^*}D^*_t=1+\int_0^t \;^{o,\bQ^*}G_sdB_s.
\]
From assumption (A3), $ \;^{o,\bQ^*}G\equiv 0$ and we obtain  $ \;^{o,\bQ^*}D^*\equiv 1$, hence (\ref{LF}) holds true. Consequently,  $X$  is also a  $(\bF,\bP)$-local martingale (by using Girsanov's theorem in the filtration $\bF$). 

In the larger filtration $\bG$ however, using the Lenglart's extension of the Girsanov's theorem, we obtain that the following is a martingale under $\bP$:
\begin{equation*}
M:=X-X_0-\int_0^t\frac{d\langle X,D^*\rangle_s}{D^*_s}= X-X_0-\int_0^{t}\frac{G_sF_s }{D^*_s}ds.
\end{equation*}
Because from (\ref{XFQ}) $\langle X\rangle =\int F^2ds$, we can write 
\begin{equation*}
X=X_0+M+\int_0^t \alpha_sd\langle M,M\rangle_s
\end{equation*}
with 
\begin{equation}\label{alpha}
\alpha_t:=\ind_{F_t\neq 0} \ind_{D^*_t\neq 0}\frac{G_t}{D^*_tF_t}.
\end{equation}  For all $t\geq 0$:
\begin{align*}
\E_\bP[\alpha_t|\F_t]&=\ind_{F_t\neq 0}\E_\bP\left [ \ind_{D^*_t\neq 0}\frac{G_t}{D^*_t}|\F_t\right ] \frac{1}{F_t}=\ind_{F_t\neq 0}\frac{\E_{\bQ^*}\left[\frac{d\bP}{d\bQ^*}\frac{G_t}{ D^*_t}|\F_t\right]}{\E_{\bQ^*}\left[\frac{d\bP}{d\bQ^*}|\F_t\right]F_t} \\
&=\ind_{F_t\neq 0} \frac{\E_{\bQ^*}\left[G_t|\F_t\right]}{\E_{\bQ^*}\left[D^*_t|\F_t\right]F_t}=0,
\end{align*}
thus, $\alpha$ has a null $(\bF,\bP)$ optional projection. The claim in (a) is thus proved.
\bigskip\\
\textit{Proof of claims (b) and (c)}.   
We write $D^*$ as product of orthogonal  $(\bG,\bQ^*)$ local martingales in two different ways. First, we write $D^*=LZ$ with
 \begin{align*}
d L_t &= L_t \frac{G_t}{D^*_t}\ind_{\{F_t\neq 0\}\cap\{D^*_t>0\}}dB_t\\
&= L_t\alpha_tF_tdB_t\\
L_0&=1
 \end{align*} 
 and 
  \begin{align*}
d  Z_t&=Z_t\frac{G_t}{D^*_t}\ind_{\{F_t= 0\}\cap \{D^*_t>0\}}dB_t\\
Z_0&=1
  \end{align*}
Second, we write $D^*=L^{(+)} Z^{(-)}$ with
 \begin{align*}
d L^{(+)}_t &= L^{(+)}_t \frac{(G_t)^+}{D^*_t}\ind_{\{F_t\neq 0\}\cap\{D^*_t>0\}}dB_t\\
&= L^{(+)}_t (\alpha_t)^+F_tdB_t\\
L^{(+)}_0 &=1
 \end{align*} and 
 \begin{align*}
d  Z^{(-)}_t&=Z^{(-)}_t \frac{\ind_{\{D^*_t>0\}} }{D^*_t}\left((G_t)^+\ind_{\{F_t= 0\}}- (G_t)^- \right)dB_t\\
Z^{(-)}_0 &=1.
  \end{align*}
  The following relations can be easily verified (using Ito's formula and the expression of the processes $X$ and  $\alpha$ in (\ref{Xrepr}) resp (\ref{alpha})):
 \begin{align*}
R_t &=\frac{1}{L_t}\\
R^{(+)}&=\frac{1}{L^{(+)}_t}.
 \end{align*}
We denote:
 \begin{align}
\label{t0}  T_0:&=\inf\{t\geq 0:D^*_t= 0\}\wedge T\\
\label{ta}\tau:&=\inf\{t\geq 0:L_t= 0\}\wedge T\\
\label{ta+} \tau^{(+)}:&=\inf\{t\geq 0:L^{(+)}_t= 0\}\wedge T.
\end{align}

We notice that the processes $L$, $L^{(+)}$, $Z$ and $Z^{(-)}$ are  such  that the measure $d\langle L\rangle$  is singular with respect to the measure $d\langle Z\rangle$ and also  the measure $d\langle L^{(+)}\rangle$  is singular with respect to the measure $d\langle Z^{(-)}\rangle$.  These local martingales being continuous, their constancy intervals are determined by the constancy intervals of their sharp bracket; hence $L$ and $Z$ do not co-move and $L^{(-)}$ and $Z^{(+)}$ do not co-move.  Additionally the process $D^*$, a nonnegative martingale, is absorbed at the level 0, once this level attained, and consequently the processes $L$, $L^{(+)}$, $Z$ and $Z^{(-)}$  are also stopped at $T_0$. This shows that the following equalities hold $\bQ^*$ a.s.: 
\begin{align}\label{tauequiv}
\{\tau<T\}&=\{L_T=0\}=\{Z_T>0\}\cap\{T_0<T\};\\\label{tauplusequiv}
\{\tau^{(+)}<T\}&=\{L^{(+)}_T=0\}= \{Z^{(-)}_T>0\}\cap\{T_0<T\}.
 \end{align} 

We know that $R_t =\frac{1}{L_t}$. Therefore:
\begin{align*}
\E_\bP\left[R_T\right]&=\E_{\bQ^*}\left[D^*_T\frac{1}{L_T}\right]=\E_{\bQ^*}\left[D^*_T\frac{1}{L_T}\ind_{L_T>0}\right]=\E_{\bQ^*}\left[Z_T\ind_{L_T>0}\right]\\
&=\E_{\bQ*}[Z_T]-\E_{\bQ^*}\left[Z_T\ind_{L_T=0}\right]
\end{align*}
and from this we deduce that
\begin{itemize}
\item[-]  if   $\bQ^*(\tau<T)>0$ we have $\E_\bP\left[R_T\right]<1$ i.e., $R$ is a strict  $(\bG,\bP)$-local martingale. The strict inequality is justified using (\ref{tauequiv}).
 Indeed, these equalities imply that 
 $$\bQ^*(\tau<T)=\bQ^*(L_T=0)>0$$ and  $\{L_T=0\}\subset \{Z_{T}>0\} $ which leads to $\E_\bQ^*\left[Z_T\ind_{L_T=0}\right]>0$, hence the result.
\item[-]  if   $\bQ^*(\tau<T)=0$ and $Z$ a true $\bQ^*$ martingale, we have $\E_\bP\left[R_T\right]=1$ i.e., $R$ is a true  $(\bG,\bP)$-martingale. 
\end{itemize}
Identical steps as above  for  $R =\frac{1}{L}$ can be used now for $R^{(+)} =\frac{1}{L^{(+)}}$. Indeed, the characterisation of $\tau^{(+)}$ in  (\ref{tauplusequiv}), can be used to prove the following:
\begin{itemize}
\item[-]  if   $\bQ^*(\tau^{(+)}<T)>0$ we have $\E_\bP\left[R^{(+)}_T\right]<1$ i.e., $R^{(+)}$ is a strict  $(\bG,\bP)$-local martingale.
\item[-]  if   $\bQ^*(\tau^{(+)}<T)=0$ and $Z^{(-)}$ a true $\bQ^*$ martingale, we have $\E_\bP\left[R^{(+)}_T\right]=1$ i.e., $R^{(+)}$ is a true  $(\bG,\bP)$-martingale. 
\end{itemize}
The claims (b) and (c) can now be proved using the above analysis and the fact that $\bQ^*(\tau^{(+)}<T)>0$ if and only if $\bQ^*(\tau<T)>0$. For more clarity, the poof of this result is stated as a separate result below (Lemma \ref{tau}). 

In particular, for (b):  As $D^*$ is  by assumption a square integrable  $(\bG,\bQ^*)$-martingale, using Proposition \ref{Y=UZ}, we deduce that, $L$, $L^{(+)}$, $Z$ and $Z^{(+)}$ are  strictly positive and uniformly integrable  $(\bG,\bQ^*)$-martingales. We use Lemma \ref{tau} below together with the above analysis to conclude. For (c) we see that there is no need to have $Z$ and $Z^{(-)}$ be true $\bQ^*$ martingales, so that no more constraints on the process $D^*$ are needed. 
   \finproof
\begin{Lemma}\label{tau} Recall the definitions of $T_0,$ $\tau$ and $\tau^{(+)}$ in equations ~(\ref{t0}),~(\ref{ta}) and ~(\ref{ta+}), respectively. The following are equivalent: (i) $\bQ^*(T_0<T)>0$;  (ii) $\bQ^*(\tau^{(+)}<T)>0$;  (iii) $\bQ^*(\tau<T)>0$.
\end{Lemma}
\proof  To begin with, we give an alternative characterisation of the sets $\{\tau<T\}$ and $\{\tau^{(+)}<T\}$, that uses the fact that  the martingale $L$ is constant on the set $\{(\omega,t): F_t(\omega)=0\}$; while $L^{(+)}$ is constant on the set $\{(\omega,t): F_t(\omega)=0,G_t(\omega)\leq 0\}$ (these properties can be easily seen from an inspection of their quadratic variation).
\begin{align*}
\{\tau<T\}&=\{T_0=\tau\}\cap\{T_0<T\}=\{F_{T_0}\neq 0\}\cap\{T_0<T\};\\
\{\tau^{(+)}<T\}&= \{T_0=\tau^{(+)}\}\cap\{T_0<T\}\\
&=\{F_{T_0}\neq 0\}\cap\{G_{T_0}>0\}\cap\{T_0<T\}.
 \end{align*} 
 Therefore, denoting by $A:=(\ind_{T_0<t})_{t\in[0,T]}$:
 \begin{align}\nonumber
 \bQ^*(\tau^{(+)}<T)&=\E_{\bQ^*}[\ind_{F_{T_0}\neq 0}\ind_{G_{T_0}>0}\ind_{T_0<T}]=\E_{\bQ^*}\left[\int_0^T\ind_{F_t\neq 0}\ind_{G_t>0}dA_t\right]\\\label{tauplus1}
 &=\E_{\bQ^*}\left[\int_0^T\ind_{F_t\neq 0}\bQ^*(G_t>0|\F_t)dA^{o,\bQ^*}_t\right]
 \end{align}
 and
 \begin{align*}
 \bQ^*(\tau<T)&=\E_{\bQ^*}[\ind_{F_{T_0}\neq 0}\ind_{T_0<T}]=\E_{\bQ^*}\left[\int_0^T\ind_{F_t\neq 0}dA_t\right]. \end{align*}
 We notice that the measure $dA_t$ does not charge the set $\{(\omega,t):G_t(\omega)= 0\}$, that is:
 $$
 dA_t=\ind_{G_t\neq 0}dA_t.
 $$
 This property can be seen as following from  Dubins-Schwarz theorem, which allows to write  
 \[
 T_0=\inf\left\{t:W_{\int_0^tG_s^2ds}=0\right\}
 \]
  where $W$ is a Brownian motion (alternatively,  the hitting time of 0 by $D^*$ cannot occur on the constancy intervals of the process $D^*$).

  Hence:
 \begin{align}\nonumber
 \bQ^*(\tau<T)&=\E_{\bQ^*}\left[\int_0^T\ind_{F_t\neq 0}\ind_{G_t\neq 0}dA_t\right]\\\label{tau1}
 &=\E_{\bQ^*}\left[\int_0^T\ind_{F_t\neq 0}\bQ^*(G_t\neq0|\F_t)dA^{o,\bQ^*}_t\right]
 \end{align}
We remark that because $\E_{\bQ^*}[G_t|\F_t]=0$, we have that $ \bQ^*(G_t<  0|\F_t)>0$ if and only if $ \bQ^*(G_t>  0|\F_t)<0$. Therefore:
 \begin{equation}\label{probaeqsets}
  \bQ^*(G_t\neq  0|\F_t)>0\text{ if and only if } \bQ^*(G_t>  0|\F_t)>0. 
 \end{equation}  
 
From  (\ref{tauplus1}), we see that $\bQ^*(\tau^{(+)}<T)>0$ if and only if the measure $dA^{o,\bQ^*}_t$ charges a subset of $$\mathcal A^{(+)}:= \{(\omega,t)\in\Omega\times [0,T]|F_t(\omega)\neq 0; \bQ^*(G_t>  0|\F_t)(\omega)>0 \}.$$
From  (\ref{tau1}), we see that $\bQ^*(\tau<T)>0$ if and only if the measure $dA^{o,\bQ^*}_t$ charges a subset of 
$$\mathcal A:=\{(\omega,t)\in\Omega\times [0,T]|F_t(\omega)\neq 0; \bQ^*(G_t\neq  0|\F_t)(\omega)>0 \}.$$

But because of (\ref{probaeqsets}) we have $\mathcal A^{(+)}=\mathcal A$. Thus, we have proved that $\bQ^*(\tau^{(+)}<T)>0$ if and only if $\bQ^*(\tau<T)>0$. Obviously, each one of these implies $\bQ^*(T_0<T)>0$.
 \finproof

\begin{rem}\label{rem}The results in Proposition \ref{ReprXfromQ} remain true if we replace the Brownian motion $B$ by the continuous martingale $X$ having the predictable representation property in a filtration $\widetilde \bF$ under $\bQ^*$. More exactly, the statements in Proposition \ref{ReprXfromQ} are true if we can replace the assumptions (A1), (A2), (A3) with the assumptions (A1'), (A2), (A3'), where
\begin{itemize}
\item[(A1')] There is a continuous martingale $X$ having the predictable representation property with respect to $(\bF^0,\bQ^*)$.
\item[(A3')] There is a nonnegative $( \bG^0,\bQ^*)$ martingale having the representation:
\begin{equation}\label{defDbis}
D^*_t=1+\int_0^t G_sdX_s,
\end{equation}
where $G$ is a $\bG$ predictable process satisfying  $\E_{\bQ^*}[G_t|\F^0_t]=0$ for all $t\in[0,T]$. 
\end{itemize}
The proof is identical; we simply have to consider that $F\equiv 1$, where $F$ is the process appearing in the proof of Proposition \ref{ReprXfromQ}, see equation (\ref{XFQ}) and to replace $B$ by $X$ everywhere. 
\end{rem}

\proof (of Theorem \ref{ThmH})
 We consider  the probability space is $(\Omega,\G,\bP)$. We assume Hypotheses \ref{H1} and \ref{H2} together with conditions (C) and $\bF$-(PRP).  Additionally, we suppose that   $ \mathcal p_T(\bF,\bG,\bP)\neq \emptyset$. By Theorem \ref{ThmQstar}, there is also a probability  $\bQ^*\in \mathcal P_T(\bF,\bG,\bP)$, such that  $D^*_t:=\E_{\bQ^*}[\frac{d\bP}{d\bQ^*}| \G_{t\wedge T}]$  has the representation introduced in (\ref{defD}). Using the Remark \ref{rem} and Proposition \ref{ReprXfromQ}, we deduce that: 
\begin{itemize}
\item[-] If $\bP\sim\bQ^*$ and $D^*_T\in L^2(\Omega,\F,\bQ^*)$, then, both $R$ and $R^{(+)}$ are true  $(\bG,\bP)$ martingales. This proves (b).
\item[-] If $\bP\ll\bQ^*$ but not equivalent,  then, both $R$ and $R^{(+)}$ are strict  $(\bG,\bP)$ local martingales. This proves (a).
\end{itemize}
\finproof
\section{Conclusions}
This paper represents a first step in analysing the potential of arbitrage arising from inside information, in presence of short sales prohibitions. Our analysis may be extended in several directions. 
 
The concept of minimal supermartingale measure that we have introduced is independent of the two filtrations setting with inside information. Hence, it may be used for studying the effect of short-sales restrictions in a general market model, without inside information. An important such direction is the study of  local risk minimisation for investors in presence of short sales restrictions. For this purpose, the connections between minimal martingale measures and minimal supermartingale measures may be interesting to further analyse. 

The two filtration setting that we used for the analysis of the inside information is very general; the study of particular enlargements of filtrations, or filtering models can possibly lead to additional, more specific results. We also expect analogous results to those in Section 4 to hold for multidimensional assets, when using the conditions (C) and $\bF$-(PRP). The framework developed for Proposition \ref{ReprXfromQ} is voluntarily more general than is needed for proving Theorem \ref{ThmQstar}, the filtration $\bF$ being larger that the natural filtration of the process $X$. This allows for an easier generalisation to a multidimensional setting.

Finally, we only were focused on arbitrages in the form of free lunches with vanishing risks; other definitions for arbitrages may be considered and  Hypothesis \ref{H1} may be replaced by a weaker condition,  the existence of a supermartingale measure for the common investors. 
\appendix
\renewcommand{\thesubsection}{\Alph{subsection}}
\section{Some useful results on the product of strictly positive, orthogonal  local martingales}\label{appendix}

Given a  uniformly integrable martingale $M$, that can be factorised as $M=UZ$ with strictly positive local martingales martingales $U$ and $Z$, it is of importance to know if and when $U$ and $Z$ are as well uniformly integrable martingales. The question has raised already much attention as it is connected to the existence of the minimal martingale measure. Delbaen and Shachermayer \cite{DelbScha98b} have shown that in general, it is not the case that $U$ and $Z$ are both uniformly integrable, even when one considers continuous processes. Below, we emphasise two particular situations where the implication that $U$ and $Z$ are martingales holds true, relevant  in our search for equivalent supermartingale measures, in particular the minimal supermartingale measure $\bQ^{\mathcal S}$.

In this appendix we use a filtered probability space that we label $(\Omega,\F, \bG,\bP)$ and all stochastic processes are considered adapted to it. Also, $\E$ now stands for $\E^\bP$, as no confusion can occur. 

\bp \label{Y=UZ} Suppose that $U$ and $Z$ are two strictly positive local martingales with  $\Delta U\Delta Z\equiv 0$ and such that the product $M:=UZ$ is a square integrable martingale. Then $U$ and $Z$ are uniformly integrable martingales.
\ep
 \proof
Without loss of generality we assume $U_0=Z_0=1$. 

First, let us notice that because $U$ and $Z$ are positive local martingales, they converge and $\E[U_\infty]\leq 1$,  $\E[Z_\infty]\leq 1$.

Let us define:
\[
\theta^n:=\inf\{t|U_t\leq 1/n\}=\inf\{t|Z_t\geq nM_t\}, n\geq 1.
\]

We show first that $(\theta^n)$ is a localising sequence for $Z$, that  transforms the stopped process $Z^n:=Z_{\theta^n\wedge \cdot}$ in a square integrable martingale. It is sufficient to notice:
\[
0<Z^n_{t}\leq n \max\{M^n _{t-},M^n_{t}\}.
\]

The inequality above is obvious on the set $\{t<\theta^n\}$. We justify now this inequality at $t=\theta^n<\infty$ as follows:

\begin{itemize}

\item[-] If  $\Delta M_{\theta^n}=0$, then $\Delta Z_{\theta^n}=0$, hence $Z_{\theta^n}\leq nM_{\theta^n}$ (with equality on $\{\theta^n<\infty\}$).
\item[-] If $\Delta M_{\theta^n}\neq 0$ and $\Delta Z_{\theta^n}=0$, then necessarily $\Delta M_{\theta^n}< 0$ by definition of $\theta^n$ and hence $ nM_{\theta^n}\leq Z_{\theta^n}\leq  nM_{\theta^n-}$ (the process Z was below $nM$ before the negative jump of $M$ at $\theta_n$).

\item[-] If $\Delta M_{\theta^n}\neq 0$, $\Delta Z_{\theta^n}\neq 0$, and $\Delta U_{\theta^n}= 0$, i.e., $U_{\theta^n}=1/n$, we have: $\Delta M_{\theta^n}=\frac{1}{n}\Delta Z_{\theta^n}$, i.e., $\Delta Z_{\theta^n}=n\Delta M_{\theta^n}$. Hence, $Z_{\theta^n}< nM_{\theta^n-}+\Delta Z_{\theta^n}=n(M_{\theta^n-}+\Delta M_{\theta^n})=nM_{\theta^n}$ (hence this configuration cannot occur, because in contradiction with the definition of $\theta^n$). 

\item[-] Finally, notice that the case $\Delta M_{\theta^n}\neq 0$, $\Delta Z_{\theta^n}\neq 0$, and $\Delta U_{\theta^n}\neq 0$ is excluded by the condition $\Delta U \Delta Z\equiv 0$.

\end{itemize}
This proves that  $Z$ is a locally  square integrable local martingale martingale and so is $U$, by a symmetric argument. 
We denote by $\mathcal Z$ the stable subset of locally square integrable local martingales generated by the stochastic integrals $\int H_sdZ_s$ and by $\mathcal Z^\perp$ the subspace orthogonal to $\mathcal Z$. If follows that $U\in \mathcal Z^\perp$. For any $t\geq 0$, let $\mathcal H_t$ be the $\bP$ complete $\sigma$-field generated by class of random variables  $(N_s,s\leq t)$ with $N$ being a square integrable element of  $\mathcal Z^\perp$.  Thus, we have a filtration $(\mathcal H_t)$ that satisfies the (H) hypothesis,  that is, all $(\mathcal H_t)$ martingales are $\bF$ martingales (see Br\'emaud and Yor \cite{bremaudyor}, Theorem 3 (3) which applies to our construction).    Also, we have that $U_t$ is $\mathcal H_t$ measurable (this is true for any $N_t$ where $N$ is a square integrable element of $\mathcal Z^\perp$ and by using a localisation  argument, for $U_t$ as well). 

We remark that $\theta^n$ is a $(\mathcal H_t)$ stopping time (as is a hitting time by $U$ of a non random level $1/n$). As $Z_{t\wedge \theta^n}$ is a square integrable martingale orthogonal to all $(\mathcal H_t)$ square integrable martingales, and the hypothesis (H) holds, we have that $\;^oZ=1$ and consequently $\E[Z_{\theta^n}|\mathcal H_{\theta^n}]=Z_0=1$. We obtain:
\[
1=\E[M_{\theta^n}]=\E[U_{\theta^n}Z_{\theta^n}]=\E\left [U_{\theta^n}\E[Z_{\theta^n}|\mathcal H_{\theta^n}]\right]=\E\left [U_{\theta^n}\right],
\]
and therefore $\theta^n$ is also a localising sequence for $U$.
Further, for $n\geq 1$:
\begin{align*}
1=\E[U_{\theta^n}]=&\E[U_\infty\ind_{\theta^n=\infty}+U_{\theta^n}\ind_{\theta^n<\infty}]\\
\leq& \E[U_\infty\ind_{\theta^n=\infty}+1/n \ind_{\theta^n<\infty}]\\
\leq & \E[U_\infty\ind_{\theta^n=\infty}]+\bP[\theta^n<\infty].
\end{align*}
It follows:
\[
\E[(U_0-U_\infty)\ind_{\theta^n=\infty}]\leq 0.
\]Since $|U_0-U_\infty|\ind_{\theta^n=\infty}\leq |U_0-U_\infty|\in \mathbb L^1,$ by dominated convergence, $\E[(U_0-U_\infty)]\leq 0$. But as a supermartingale, we also have $\E[(U_0-U_\infty)]\geq 0$. Hence $U$ is a uniformly integrable martingale.

The roles of $U$ and $Z$ being interchangeable, a symmetric reasoning can be made to show that $Z$ is a uniformly integrable martingale. 

\finproof

\bl\label{M=UZ(2)}
Suppose that $M=\mathcal E (L)$ is a strictly positive and uniformly integrable martingale. We consider a c\`adl\`ag predictable process $H$ with state space $\{0,1\}$  and $U:=\mathcal E(\int HdL)$ and $Z:=\mathcal E(\int (1-H)dL)$. Then,  if $H$ has  bounded total variation, then both $U$ and $Z$ are uniformly integrable. 
\el
\proof 
Without loss of generality we shall assume hereafter that $H_0=0$. 

Let $T^0=0$ and for $k\geq 1$, $T^k:=\inf\{t\geq T^{k-1}\;|H_t\neq H_{T^{k-1}}\; \}$ i.e., the successive jumping times of the process $H$, and let $N=(N_t)_{t\geq 0}$ be the counting process associated with $(T^k)_{k\in\mathbb N}$. We have that $H$ has bounded variation if and only if $N$ is bounded.   We denote $T^\infty=\lim_{k\to\infty}T^k$.  Also we denote:
 \[
 M^{(k)}:=\mathcal E(L_{T^k\wedge \cdot}-L_{T^{k-1}\wedge \cdot}),\quad k=1,2,....
 \]We notice that $M^{(k)}$ are orthogonal local  martingales and:
 \begin{align*}
 M=\Pi_{k\geq 1}M^{(k)}\quad  U=\Pi_{k\geq 1} M^{(2k)},\quad  Z=\Pi_{k\geq 1} M^{(2k-1)}.
 \end{align*}  
 Also $M^{(1)}=M_{T^1\wedge\cdot}$ is a uniformly integrable martingale. So is every $M^{(k)}$, $k>1$:
 \be\label{jump1}
\E[M^{(k+1)}_\infty]=\E\left[\E(M^{(k+1)}_\infty|\G_{T^k})\right]=\E\left[\E(M_{T^{k+1}}|\G_{T^k})\frac{1}{M_{T^k}}\right]=1.
\ee
 We will now prove the stated properties for $U$; identical arguments hold for $Z$. We denote
 \[
 U_\infty:=\lim_{t\to\infty}U_t=U_{T^\infty}.
 \]
 We notice that $U_{T^{2k}}=U_{T^{2k-2}}M^{(2k)}_\infty$ and $\E[U_{T^{2k}}|\F_{T_{2k-2}}]=U_{T^{2k-2}}M^{(2k)}_\infty$, therefore, for all $k\geq 1$:
 \[
 \E[U_{T^{2k}}]=\E[M^{(2k)}_\infty U_{T^{2k-2}}]=\E[U_{T^{2k-2}}]=\cdots= U_0=1.
 \]
This proves that if $T^\infty=+\infty$ $a.s.$, i.e. $(N_t)$ is finite, then $(T^{2k})$ is a localising sequence for $U$.

Therefore,  whenever $(N_t)$ is bounded, then $U$ is a uniformly integrable martingale.  
\finproof

\renewcommand{\refname}{References}

\end{document}